\documentclass[10pt]{article}

\usepackage{amsmath,amsthm,amsfonts,mathrsfs,anysize,fancyhdr,epsfig,mdframed, float,graphicx,umoline,dsfont}

\usepackage{amsmath}                                                        
\usepackage{graphicx}                                                       
\usepackage{subfigure}
\usepackage{enumitem}                                                    
\usepackage{hyperref}
\usepackage{amssymb}                                                        
\usepackage[mathscr]{eucal}                                             
\usepackage{cancel}                                                             
\usepackage[normalem]{ulem}                                                                 
\usepackage{pstricks}
\usepackage{rotating}
\usepackage{lscape}
\usepackage[paperwidth=8.5in,paperheight=11in,top=1.25in, bottom=1.25in, left=1.00in, right=1.00in]{geometry}
\usepackage{mathtools}                                                      
\mathtoolsset{showonlyrefs=true}                                    
\usepackage{fixltx2e,amsmath}                                           
\MakeRobust{\eqref}

\usepackage{mathdots}
\usepackage{amsthm}                                                             
\allowdisplaybreaks                                                             
\theoremstyle{plain}
\newtheorem{theorem}{Theorem}

\theoremstyle{definition}

\newtheorem{example}[theorem]{Example}
\newtheorem{remark}[theorem]{Remark}

\DeclarePairedDelimiter{\ceil}{\lceil}{\rceil}
\DeclarePairedDelimiter\floor{\lfloor}{\rfloor}

\def \y {{\eta}}
\def \s {{\sigma}}
\def \d {{\delta}}
\def \g {{\gamma}}
\def \a {{\alpha}}

\def \xbar {\bar{x}}

\def \R  {{\mathbb {R}}}
\def \x {{\xi}}
\def \g {{\gamma}}
\def \e {{\varepsilon}}

\def \t {{\tau}}

\renewcommand{\]}{\right]}

\newcommand\N{\mathbb{N}}

\def \caratt {{\mathds{1}}}

\newcommand\Eb{E}

\newcommand\Rb{\mathbb{R}}

\newcommand\G{\Gamma}
\newcommand\Gh{\hat{\Gamma}}

\newcommand\gam{\gamma}

\def \phi {{\varphi}}

\begin{document}

\title{Pricing Bermudan options under local L\'evy models with default}

\author{Anastasia Borovykh\thanks{Dipartimento di Matematica, Universit\`a di Bologna, Bologna, Italy.
\textbf{e-mail}: anastasia.borovykh2@unibo.it} \and Andrea Pascucci\thanks{Dipartimento di Matematica, Universit\`a di
Bologna, Bologna, Italy. \textbf{e-mail}: andrea.pascucci@unibo.it} \and Cornelis W. Oosterlee\thanks{Centrum Wiskunde \& Informatica, Amsterdam, The Netherlands. \textbf{e-mail}: c.w.oosterlee@cwi.nl} \thanks{Delft University of Technology, Delft, The Netherlands.}}

\date{This version: \today}

\maketitle

\begin{abstract}
We consider a defaultable asset whose risk-neutral pricing dynamics are described by an
exponential L\'evy-type martingale. This class of models allows for a local
volatility, local default intensity and a locally dependent L\'evy measure. We present a pricing
method for Bermudan options based on an analytical approximation of the characteristic function
combined with the COS method. Due to a special form of the obtained characteristic function the price can
be computed using a Fast Fourier Transform-based algorithm resulting in a fast and accurate
calculation. The Greeks can be computed at almost no additional computational cost. Error bounds for the approximation of the characteristic function as well as for the
total option price are given.
\end{abstract}


%
%

\section{Introduction}

In financial mathematics, the fast and accurate pricing of financial derivatives is an important
branch of research. Depending on the type of financial derivative, the mathematical task is
essentially the computation of integrals, and this sometimes needs to be performed in a recursive
way in a time-wise direction. For many stochastic processes that model the financial
assets, these integrals can be most efficiently computed in the Fourier domain. However, for some
relevant and recent stochastic models the Fourier domain computations are not at all
straightforward, as these computations rely on the availability of the characteristic function of
the stochastic process (read: the Fourier transform of the transitional probability distribution),
which is not known. This is especially true for state-dependent asset price processes, and for
asset processes that include the notion of default in their definition. With the derivations and
techniques in the present paper we make available the highly efficient pricing of so-called
Bermudan options to the above mentioned classes of state-dependent asset dynamics, including jumps
in asset prices and the possibility of default. In this sense, the class of asset models for which
Fourier option pricing is highly efficient increases by the contents of the present paper.
Essentially, we approximate the characteristic function by an advanced Taylor-based expansion in
such a way that the resulting characteristic function exhibits favorable properties for the
pricing methods.

Fourier methods have often been among the winners in option pricing competitions such as BENCHOP \cite{benchop}. In \cite{FangO09}, a Fourier method called the COS method, as introduced in \cite{FangO08}, was extended to the pricing of Bermudan options. The
computational efficiency of the method was based on a specific structure of the characteristic
function allowing to use the fast Fourier transform (FFT) for calculating the continuation value of the option. Fourier
methods can readily be applied to solving problems under asset price dynamics for which the
characteristic function is
available. This is the case for exponential L\'evy models, such as the Merton model developed in
\cite{merton1976option}, the Variance-Gamma model developed in \cite{vg}, but also for the Heston
model \cite{heston1993}. However, in the case of local volatility, default and state-dependent jump
measures there is no closed form characteristic function available and the COS method can not be
readily applied.

Recently, in \cite{pascucci-riga} the so-called {\it adjoint expansion method} for the approximation of
the characteristic function in local L\'evy models is presented. This method is worked out in the
Fourier space by considering the adjoint formulation of the pricing problem, that is using a
backward parametrix expansion as was also later done in \cite{BallyKohatsu-Higa}. In this paper we
generalize this method to include a defaultable asset whose risk-neutral pricing dynamics are
described by an exponential L\'evy-type martingale with a state-dependent jump measure, as has
also been considered in \cite{LorigPP2015} and in \cite{JacquierLorig2013}. 

Having obtained the
analytical approximation for the characteristic function we combine this with the COS method for
Bermudan options. We show that this analytical formula for the characteristic function still
possesses a structure that allows the use of a FFT-based method in order to calculate the
continuation value. This results in an efficient and accurate computation of the Bermudan option
value and of the Greeks. The
characteristic function approximation used in the COS method is already very accurate for the
$2$nd-order approximation, meaning that the explicit formulas are simple and this makes method
easy and quick to implement. Finally, we present a theoretical justification of the accurate
performance of the method by giving the error bounds for the approximated characteristic function.

The rest of this paper is organized as follows. In Section \ref{section1} we present the general
framework which includes a local default intensity, a state-dependent jump measure and a local
volatility function. Then we derive the adjoint expansion of the characteristic function. In
Section \ref{section2} we propose an efficient algorithm for calculating the Bermudan option
value, which makes use of the Fast Fourier transform. In Section \ref{section3} we prove error
bounds for the $0$th- and $1$st-order approximation, justifying the accuracy of the method.
Finally, in Section \ref{section4} numerical examples are presented, showing the flexibility,
accuracy and speed of the method.

\section{General framework} \label{section1}
We consider a defaultable asset $S$ whose risk-neutral dynamics are given by:
 \begin{align}
 S_t &= \caratt_{\{t<\zeta\}}e^{X_t},\nonumber\\
 dX_t &= \mu (t,X_t)dt+\sigma (t,X_t)dW_t+\int_\mathbb{R}d\tilde   N_t(t,X_{t-},dz)z,\nonumber \\
d\tilde N_t(t,X_{t-},dz) &= dN_t(t,X_{t-},dz)-\nu (t,X_{t-},dz)dt,\nonumber\\
 \zeta &= \inf\{t\geq 0 :\int_0^t\gamma(s,X_{s})ds\geq \varepsilon\}\label{eq:hetmodel},
 \end{align}
where $\tilde N_t(t,x,dz)$ is a compensated 
random measure with state-dependent L\'evy measure $\nu(t,x,dz)$. The default time $\zeta$ of $S$
is defined in a canonical way as the first arrival time of a doubly stochastic Poisson process
with local intensity function $\gamma(t,x)\geq 0$, and $\varepsilon \sim \mathrm{Exp}(1)$ and is
independent of $X$. Thus the model features:
\begin{itemize}
\item a local volatility function $\sigma(t,x)$;
\item a local L\'evy measure: jumps in $X$ arrive with a state-dependent intensity
described by the local L\'evy measure $\nu(t,x,dz)$. The jump intensity and jump distribution can
thus change depending on the value of $x$. A state-dependent L\'evy measure is an important
feature because it allows to incorporate stochastic jump-intensity into the modeling framework;
\item a local default intensity $\g(t,x)$: the asset $S$ can default with a state-dependent default intensity.
\end{itemize}
This way of modeling default is also considered in a diffusive setting in \cite{JDCEV} and for
exponential L\'evy models in \cite{capponi}.

We define the filtration of the market observer to be $\mathcal{G}=\mathcal{F}^X\vee
\mathcal{F}^D$, where $\mathcal{F}^X$ is the filtration generated by $X$ and
$\mathcal{F}_t^D:=\sigma(\{\zeta\leq u\},u\leq t)$, for $t\ge0$, is the filtration of the default.
We assume
\begin{equation}\label{nusomm}
  \int_\mathbb{R}e^{|z|}\nu(t,x,dz)<\infty,
\end{equation}
and by imposing that the discounted asset price $\tilde S_t := e^{-rt}S_{t}$ is a
$\mathcal{G}$-martingale, we get the following restriction on the drift coefficient:
 $$\mu(t,x) = \gamma(t,x)+r-\frac{\sigma^2(t,x)}{2}-\int_\mathbb{R}\nu(t,x,dz)(e^z-1-z).$$
Is it well-known (see, for instance, \cite[Section 2.2]{linetsky2006bankruptcy}) that the price
$V$ of a European option with maturity $T$ and payoff $\Phi(S_{T})$  is given by
\begin{align}\label{e1}
 V_{t} = \caratt_{\{\zeta>t\}} e^{-r(T-t)}\Eb \[e^{-\int_t^T \gam(s,X_s) ds}\phi(X_T)  | X_t \],\qquad t\le T,
\end{align}
where $\phi(x)=\Phi(e^{x})$.
Thus, in order to compute the price of an option, we must evaluate functions of the form
\begin{align}\label{expectation}
 u(t,x):= \Eb \[e^{-\int_t^T \gam(s,X_s) ds}\phi(X_T)| X_t = x \] .
\end{align}
Under standard assumptions, $u$ can be expressed
as the classical solution of the following Cauchy problem
\begin{align}\label{eq:v.pide}
&\begin{cases}
 L u(t,x)=0,\qquad & t\in[0,T[,\ x\in\mathbb{R}, \\
 u(T,x) =  \phi(x),& x \in\mathbb{R},
\end{cases}
\end{align}
where $L$ is the integro-differential operator
 \begin{align}\label{opL}
 L u(t,x) =&\ \partial_t u(t,x)+r\partial_xu(t,x) +\gamma(t,x)(\partial_x u(t,x)-u(t,x))+\frac{\sigma^2(t,x)}{2}(\partial_{xx}-\partial_x)u(t,x)\nonumber\\
 &-\int_\mathbb{R}\nu(t,x,dz)(e^z-1-z)\partial_x u(t,x)+\int_\mathbb{R}\nu(t,x,dz)(u(t,x+z)-u(t,x)-z\partial_x u(t,x)).
 \end{align}
The function $u$ in \eqref{expectation} can be represented as an integral with respect to the
transition distribution of the defaultable log-price process $\log S$:
\begin{align}\label{eq:v.def1}
 u(t,x) = \int_\Rb  \phi(y)\G(t,x;T,dy).
\end{align}
Here we notice explicitly that $\G(t,x;T,dy)$ is not necessarily a standard probability measure
because its integral over $\R$ can be strictly less than one; nevertheless, with a slight abuse of
notation, we say that its Fourier transform
 $$\Gh(t,x;T,\x):=\mathcal{F}(\Gamma(t,x;T,\cdot))(\xi):= \int_\mathbb{R}e^{i\xi y}\Gamma(t,x;T,dy),\qquad \x\in\R,$$
is the characteristic function of $\log S$.

\subsection{Adjoint expansion of the characteristic function}
In this section we generalize the results in \cite{pascucci-riga} to our framework and
develop an expansion of the coefficients
 $$a(t,x):=\frac{\sigma^2(t,x)}{2},\qquad \gamma(t,x),\qquad \nu(t,x,dz),$$
around some point $\bar{x}$. The coefficients $a(t,x)$, $\gamma(t,x)$ and $\nu(t,x,dz)$ are
assumed to be continuously differentiable with respect to $x$ up to order $N\in\mathbb{N}$.

From now on for simplicity we assume that the coefficients are independent of $t$ (see Remark
\ref{r1} for the general case). First we introduce the $n$th-order approximation of $L$ in
\eqref{opL}:
\begin{align}
 L_n =&\ L_0 + \sum_{k=1}^n\Big( (x-\bar x)^k a_k (\partial_{xx}-\partial_x)+(x-\bar x)^k\gamma_k\partial_x-(x-\bar x)^k\gamma_k\nonumber\\
 & -\int_\mathbb{R}(x-\bar x)^k\nu_k(dz)(e^z-1-z)\partial_x +\int_\mathbb{R}(x-\bar x)^k\nu_k(dz)(e^{z\partial_x}-1-z\partial_x)\Big),
 \end{align}
where
 \begin{align}
 L_0 &=  \partial_t +r\partial_x+ a_0 (\partial_{xx}-\partial_x)+ \gamma_0\partial_x-\gamma_0
  -\int_\mathbb{R}\nu_0(dz)(e^z-1-z)\partial_x
  +\int_\mathbb{R}\nu_0(dz)(e^{z\partial_x}-1-z\partial_x),
 \end{align}
and
\begin{align}
  a_k= \frac{\partial_x^k a(\bar x)}{k!},\qquad
  \gamma_k = \frac{\partial_x^k \gamma(\bar x)}{k!},\qquad
  \nu_k(dz) = \frac{\partial_x^k \nu (\bar x,dz)}{k!},\qquad\ k\ge 0.
 \end{align}
The basepoint $\bar x$ is a constant parameter which can be chosen freely. In general the simplest
choice is $\bar x = x$ (the value of the underlying at initial time $t$): we will see that in this
case the formulas for the Bermudan option valuation are simplified.

Let us assume for a moment that $L_{0}$ has a fundamental solution $G^{0}(t,x;T,y)$ that
is defined as the solution of the Cauchy problem
 $$\begin{cases}
 L_0 G^{0}(t,x;T,y) =0\qquad & t\in[0,T[,\ x\in\mathbb{R}, \\
 G^{0}(T,\cdot;T,y) =\d_{y}. 
 \end{cases}$$
In this case we define the $n$th-order approximation of $\Gamma$ as
 $$\Gamma^{(n)}(t,x;T,y) = \sum_{k=0}^n G^{k}(t,x;T,y),$$
where, for any $k\ge 1$ and $(T,y)$, $G^{k}(\cdot,\cdot;T,y)$ is defined recursively through the
following Cauchy problem
 $$
  \begin{cases}
 L_0 G^{k}(t,x;T,y) = -\sum\limits_{h=1}^k(L_h-L_{h-1})G^{k-h}(t,x;T,y)\qquad & t\in[0,T[,\ x\in\mathbb{R}, \\
 G^{k}(T,x;T,y) =0,& x \in\mathbb{R}.
 \end{cases}
 $$
Notice that
 \begin{align}
 L_h-L_{h-1}=&\ (x-\bar x)^h a_h (\partial_{xx}-\partial_x)+(x-\bar x)^h\gamma_h\partial_x-(x-\bar x)^h\gamma_h\\
 & -\int_\mathbb{R}(x-\bar x)^h\nu_h(dz)(e^z-1-z)\partial_x +\int_\mathbb{R}(x-\bar x)^h\nu_h(dz)(e^{z\partial_x}-1-z\partial_x).
 \end{align}
Correspondingly, the $n$th-order approximation of the characteristic function $\Gh$ is defined to
be
\begin{equation}\label{adapprox}
 \Gh^{(n)}(t,x;T,\xi)=\sum_{k=0}^n \mathcal{F}\left(G^{k}(t,x;T,\cdot)\right)(\xi):=\sum_{k=0}^n \hat G^{k}(t,x;T,\xi),\qquad \x\in\R.
\end{equation}
Now we remark that the operator $L$ acts on $(t,x)$ while the characteristic function is a Fourier
transform taken with respect to $y$: in order to take advantage of such a transformation, in the
following theorem we characterize $\Gh^{(n)}$ in terms of the Fourier transform of the adjoint
operator $\tilde L=\tilde L^{(T,y)}$ of $L$, acting on $(T,y)$.
\begin{theorem}[Dual formulation] For any $(t,x)\in ]0,T]\times\mathbb{R}$, the function $G^{0}(t,x;\cdot, \cdot)$ is defined through the following dual
Cauchy problem
\begin{equation}\label{eq:cauchy1}
  \begin{cases}
 \tilde L_0^{(T,y)} G^{0}(t,x;T,y) =0\qquad & T>t,\ y\in\mathbb{R}, \\
 G^{0}(T,x;T,\cdot) =\d_{x}. 
 \end{cases}
\end{equation}
where
\begin{align}
 \tilde L_0^{(T,y)}&=  -\partial_T -r\partial_y+ a_0 (\partial_{yy}+\partial_y)- \gamma_0\partial_y-\gamma_0
 +\int_\mathbb{R}\nu_0(dz)(e^z-1-z)\partial_y +\int_\mathbb{R}\bar \nu_0(dz)(e^{z\partial_y}-1-z\partial_y).
 \end{align}
Moreover, for any $k\geq 1$, the function $G^{k}(t,x;\cdot, \cdot)$ is defined through the dual Cauchy problem as follows:
\begin{equation}\label{eq:cauchy2}
  \begin{cases}
 \tilde L_0^{(T,y)} G^{k}(t,x;T,y) = -\sum\limits_{h=1}^k\left(\tilde L_h^{(T,y)}-\tilde L_{h-1}^{(T,y)}\right)
 G^{k-h}(t,x;T,y)\qquad & T>t,\ y\in\mathbb{R}, \\
 G^{k}(T,x;T,y) =0& y\in\mathbb{R},
 \end{cases}
\end{equation}
with
 \begin{align}
 \tilde L_h^{(T,y)}-\tilde L_{h-1}^{(T,y)} =&\ a_h h(h-1)(y-\bar x)^{h-2}+a_h (y-\bar x)^{h-1} \left(2h\partial_y+(y-\bar x)(\partial_{yy}+\partial_y)+h\right)\\
 &-\gamma_hh(y-\bar x)^{h-1}-\gamma_h(y-\bar x)^h\left(\partial_y+1\right)\\
 & +\int_\mathbb{R}\nu_h(dz)(e^z-1-z)\left(h(y-\bar x)^{h-1}+(y-\bar x)^h\partial_y\right) \nonumber\\
 &+\int_\mathbb{R}\bar\nu_h(dz)\left((y+z-\bar x)^he^{z\partial_y}-(y-\bar x)^h-z\left(h(y-\bar x)^{h-1}-(y-\bar
 x)^h\partial_y\right)\right),
 \end{align}
where in defining the adjoint of the operator we use the notation
\begin{align}
 e^{z\partial_y}f(y) :=\sum_{n=0}^\infty\frac{z^n}{n!}\partial_y^n f(y)=f(y+z).
\end{align}
\end{theorem}
\noindent
Notice that the adjoint Cauchy problems \eqref{eq:cauchy1} and \eqref{eq:cauchy2} admit a
solution in the Fourier space and can be solved explicitly; in fact, we have
 $$\mathcal{F}\left(\tilde L_0^{(T,\cdot)}G^{k}(t,x;T,\cdot)\right)(\xi) = \psi(\xi)\hat G^{k}(t,x;T,\x)-\partial_T\hat G^{k}(t,x;T,\x),$$
where $ \psi(\xi)$ is the characteristic exponent of the L\'evy process with coefficients
$\gamma_0$, $a_0$ and $\nu_0(dz)$, that is
\begin{equation}\label{charexpp}
 \psi(\xi) = i\xi(r+\gamma_0)
 +a_0(-\xi^2-i\xi)-\gamma_0-\int_\mathbb{R}\nu_0(dz)(e^z-1-z)i\xi+\int_\mathbb{R}\nu_0(dz)(e^{iz\xi}-1-iz\xi).
\end{equation}
Thus the solution (in the Fourier space) to problems \eqref{eq:cauchy1} and \eqref{eq:cauchy2} is
given by
\begin{equation}\label{theGk}
\begin{split}
 \hat G^{0}(t,x;T,\xi) &= e^{i\xi x+(T-t)\psi(\xi)},\\
 \hat G^{k}(t,x;T,\xi) &= -\int_t^Te^{\psi(\xi)(T-s)}\mathcal{F}\left(\sum_{h=1}^k\left(\tilde L_h^{(s,\cdot)}-\tilde
 L_{h-1}^{(s,\cdot)}\right)G^{k-h}(t,x;s,\cdot)\right)(\xi)ds,\qquad k\ge1.
\end{split}
\end{equation}
Now we consider the general framework and in particular we drop the assumption on the existence of
the fundamental solution of $L_{0}$: in this case, we define the $n$th-order approximation of the
characteristic function $\Gh$ as in \eqref{adapprox}, with $\hat{G}^{k}$ given by \eqref{theGk}.
We also notice that
\begin{align}
 &\mathcal{F}\left(\left( \tilde L_h^{(s,\cdot)}-\tilde L_{h-1}^{(s,\cdot)}\right)u(s,\cdot)\right)(\x) =\\
 &\qquad \left(a_hh(h-1)
 (-i\partial_\xi-\bar x)^{h-2}+a_h
 (-i\partial_\xi-\bar x)^{h-1} \left(-2hi\xi+(-i\partial_\xi-\bar
 x)(-\xi^2-i\xi)+h\right)\right)\hat u(s,\xi)\\
 &\qquad-\left(\gamma_hh(-i\partial_\xi-\bar x)^{h-1}-\gamma_h(-i\partial_\xi-\bar x)^h\left(i\xi-1\right) \right)\hat u(s,\xi)\\
 &\qquad +\int_\mathbb{R}\nu_h(dz)(e^z-1-z)\left(h(-i\partial_\xi-\bar x)^{h-1}-(-i\partial_\xi-\bar x)^hi\xi\right)\hat u(s,\xi) \nonumber\\
 &\qquad+\int_\mathbb{R}\nu_h(dz)\left( (-i\partial_y-z-\bar x)^h e^{i\xi z}-(-i\partial_y-\bar x)^h+z\left(h(-i\partial_\xi-\bar x)^{h-1}
 -(-i\partial_\xi-\bar x)^hi\xi\right)\right)\hat u(s,\xi).
 \end{align}

\begin{remark}\label{r1}
In case the coefficients $\gamma$, $\sigma$, $\nu$ depend on time, the solutions to the Cauchy
problems are similar:
\begin{equation}\label{Ghat}
\begin{split}
 \hat G^{0}(t,x;T,\xi) &= e^{i\xi x}e^{\int_t^T\psi(s,\xi)ds},\\
 \hat G^{k}(t,x;T,\xi) &= -\int_t^Te^{\int_s^T\psi(\tau,\xi)d\tau}\mathcal{F}\left(\sum_{h=1}^k\left(\tilde
  L_h^{(s,\cdot)}(s)-\tilde L_{h-1}^{(s,\cdot)}(s)\right)G^{k-h}(t,x;s,\cdot)\right)(\xi)ds,
\end{split}
\end{equation}
with
\begin{align}
 &\psi(s,\xi) = i\xi(r+\gamma_0(s))
 +a_0(s)(-\xi^2-i\xi)-\int_\mathbb{R}\nu_0(s,dz)(e^z-1-z)i\xi+\int_\mathbb{R}\nu_0(s,dz)(e^{iz\xi}-1-iz\xi),\\
 &\tilde L_h^{(s,y)}(s)-\tilde L_{h-1}^{(s,y)}(s) = a_h(s) h(h-1)(y-\bar x)^{h-2}+a_h(s) (y-\bar x)^{h-1} \left(2h\partial_y+(y-\bar x)(\partial_{yy}+\partial_y)+h\right)\\
 &\qquad\qquad\qquad\qquad\qquad-\gamma_h(s)h(y-\bar x)^{h-1}-\gamma_h(s)(y-\bar x)^h\left(\partial_y+1\right)\\
 & \qquad\qquad\qquad\qquad\qquad+\int_\mathbb{R}\nu_h(s,dz)(e^z-1-z)\left(h(y-\bar x)^{h-1}+(y-\bar x)^h\partial_y\right) \nonumber\\
 &\qquad\qquad\qquad\qquad\qquad+\int_\mathbb{R}\bar\nu_h(s,dz)\left((y+z-\bar x)^he^{z\partial_y}-(y-\bar x)^h-z\left(h(y-\bar x)^{h-1}-(y-\bar
 x)^h\partial_y\right)\right).
\end{align}
\end{remark}
\noindent From these results one can already see that the dependency on $x$ comes in through
$e^{i\xi x}$ and after taking derivatives the dependency on $x$ will take the form $(x-\bar x)^m
e^{i\xi x}$: this fact will be crucial in our analysis.
\begin{example}
To see the above dependency more explicitly for the second-order approximation of the characteristic function we consider, for ease of notation, a simplified model: a one-dimensional
local L\'evy model where the log-price solves the SDE
\begin{equation}\label{eq:model}
 dX_t=\mu (X_{t})dt+\sigma (X_{t})dW_t+\int_\mathbb{R}d\tilde N_t(dz)z.
\end{equation}
This model is a simplification of the original model,
since we consider only a local volatility function, and no local default or state-dependent L\'evy
measure. Thus only a Taylor expansion of the local volatility coefficient is used. However, the
dependency that we will see generalizes in the same way to the local default and state-dependent measure.
\noindent
By the martingale condition we have
 $$\mu (x) = r -a(x)-\int_\mathbb{R}\nu(dz)(e^z-1),$$
and therefore the Kolmogorov operator of \eqref{eq:model} reads
\begin{align}
 Lu(t,x)=&\ \partial_t u(t,x)+r \partial_x
 u(t,x)+a(t,x)(\partial_{xx}-\partial_x)u(t,x)\\
 &-\int_\mathbb{R}\nu(dz)(e^z-1)+\int_{\mathbb{R}}\nu(dz)\left(u(t,x+z)-u(t,x)\right).
\end{align}
In this case, we have the following explicit approximation formulas for the characteristic
function $\Gh(t,x;T,\x)$:
\begin{equation}\label{struc}
 \Gh(t,x;T,\x)\ \approx\ \Gh^{(n)}(t,x;T,\x):= e^{i\xi x + (T-t)\psi(\xi)} \sum_{k=0}^n \hat F^{k}(t,x;T,\xi),\qquad n\ge 0,
\end{equation}
with
 $$ \psi(\xi) = ir\xi -a_0(\xi^2+i\xi)-\int_\mathbb{R}\nu(dz)(e^z-1)i\xi+\int_{\mathbb{R}}\nu (dz)\left(e^{iz\xi }-1\right),$$
and
\begin{align}\label{struc2}
 \hat F^{k}(t,x;T,\xi)&=\sum_{h=0}^{k}g^{(k)}_h(T-t,\xi)(x-\bar x)^h;
\end{align}
here, for $k=0,1,2$, we have
\begin{align}
 g^{(0)}_0(s,\xi) =&\ 1,\\
 g^{(1)}_0(s,\xi) =&\ a_1 s^2 (\xi^2+i\xi )\frac{i}{2}\psi'(\xi),\\
 g^{(1)}_1(s,\xi) =&\, -a_1 s (\xi^2+i\xi ),\\
 g^{(2)}_0(s,\xi) =&\ \frac{1}{2}s^2 a_2\xi(i+\xi )\psi ''(\xi)-\frac{1}{6}s^3\xi(i+\xi)(a_1^2(i+2\xi)\psi'(\xi)-2a_2\psi'(\xi)^2+a_1^2\xi(i+\xi)\psi''(\xi))\\
 &-\frac{1}{8}s^4a_1^2\xi^2(i+\xi)^2\psi'(\xi)^2,\\
 g^{(2)}_1(s,\xi) =&\ \frac{1}{2}s^2\xi
 (i+\xi)(a_1^2(1-2i\xi)+2ia_2\psi''(\xi))-\frac{1}{2}s^3ia_1^2\xi^2(i+\xi)^2\psi''(\xi),\\
 g^{(2)}_2(s,\xi) =&\, -a_2 s \xi(i+\xi)+\frac{1}{2}s^2a_1^2\xi^2(i+\xi)^2.
\end{align}
\end{example}
\bigskip

Using the notation from above, we can write in the same way the approximation formulas for the
general case. Here we present the results for $k=0,1$, since higher-order formulas are too long to
include. For the full formula we refer to Appendix \ref{app2}. We have:
\begin{align}
g_0^{(0)}(s,\xi)=&1,\\ \label{g01}
g_0^{(1)}(s,\xi)=&\frac{i}{2}a_1s^2(\x^2+i\x)\psi'(\x)+\frac{1}{2}\gamma_1s^2(i+\x)\psi'(\x)-\frac{1}{2}\int_\mathbb{R}\nu_1(dz)(e^z-1-z)s^2\x\psi'(\x)\\
&-\frac{1}{2}\int_\mathbb{R}\nu_1(dz)(ie^{i\xi z}-i+\x z)s^2\psi'(\x),\\
g_1^{(1)}(s,\x)=&-a_1s(\xi^2+i\xi)+\gamma_1si(i+\x)-\int_\mathbb{R}\nu_1(dz)(e^z-1-z)s\x i\\
&+\int_\mathbb{R}\nu_1(dz)(e^{i\x z}-1-\x i z)s.
\end{align}
\begin{remark}\label{r3}
From \eqref{struc}-\eqref{struc2} and  \eqref{struc3}  we clearly see that the approximation of order $n$ is a function
of the form
\begin{equation}\label{struc3}
  \Gh^{(n)}(t,x;T,\x):= e^{i\xi x} \sum_{k=0}^n (x-\bar x)^k g_{n,k}(t,T,\x),
\end{equation}
where the coefficients $g_{n,k}$, with $0\le k\le n$, depend only on $t,T$ and $\x$, but not on
$x$. The approximation formula can thus always
be split into a sum of products of functions depending only on $\xi$ and functions that are linear combinations of $(x-\bar x)^m e^{i\xi x}$, $m\in\N_{0}$.
\end{remark}

\section{Bermudan option valuation}\label{section2}
A Bermudan option is a financial contract in which the holder can exercise at a predetermined
finite set of exercise moments prior to maturity, and the holder of the option receives a payoff
when exercising. Consider a Bermudan option with a set of $M$ exercise moments $\{t_1,...,t_M\}$,
with $0\le t_{1}<t_{2}<\cdots<t_{M}=T$.
When the option is exercised at time $t_{m}$ the holder receives the payoff
$\Phi\left(t_{m},S_{t_{m}}\right)$.
Recalling \eqref{e1}, the no-arbitrage value of the Bermudan option at time $t$ is
\begin{align}v\left(t,X_{t}\right)=\caratt_{\{\zeta>t\}}\sup_{\t \in \mathcal{T}_{t}}E\[e^{-\int_{t}^{\t}
\left(r+\gam(s,X_s)\right) ds}\phi(\t,X_{\t})|X_{t}\],\end{align} where $\phi(t,x)=\Phi(t,e^{x})$
and $\mathcal{T}_{t}$ is the set of all $\mathcal{G}$-stopping times taking values in
$\{t_1,...,t_M\}\cap[t,T]$. For a Bermudan Put option with strike price $K$, we simply have
$\phi(t,x)=\left(K-e^{x}\right)^{+}$. By the dynamic programming approach, the option value can be
expressed by a backward recursion as
 $$v(t_{M},x)=\caratt_{\{\zeta>t_{M}\}}\phi(t_{M},x)$$
and
\begin{equation}\label{bermud}
 \begin{cases}
 c(t,x)=
 E\left[
 e^{\int_{t}^{t_m}\left(r+\gamma(s,X_s)\right)ds}v(t_{m},X_{t_{m}})|X_{t}=x\right],\qquad &t\in[t_{m-1},t_{m}[\\
 v(t_{m-1},x)=\caratt_{\{\zeta>t_{m-1}\}}\max\{\phi(t_{m-1},x),c(t_{m-1},x)\},\qquad
 &m\in\{2,\dots,M\}.
\end{cases}
\end{equation}
In the above notation $v(t,x)$ is the option value and $c(t,x)$ is the so-called continuation
value. The option value is set to be $v(t,x) = c(t,x)$ for $t\in\, ]t_{m-1},t_m[$, and, if $t_1>0$, also for $t\in [0,t_1[$.
\begin{remark}
Since the payoff of a Call option grows exponentially with the log-stock price, this  may
introduce significant cancellation errors for large domain sizes. For this reason we price Put
options only using our approach and we employ the well-known Put-Call parity to price Calls via
Puts. This is a rather standard argument (see, for instance, \cite{bowen}).
\end{remark}

\subsection{An algorithm for pricing Bermudan Put options}\label{sec31}
The COS method proposed by \cite{FangO09} is based on the insight that the Fourier-cosine series
coefficients of $\G(t,x;T,dy)$ (and therefore also of option prices) are
closely related to the characteristic function of the underlying process, namely the following relationship holds:
 $$\int_a^b e^{i\frac{k\pi}{b-a}}\Gamma (t,x;T,dy) \approx \Gh\left(t,x;T,\frac{k\pi}{b-a}\right).$$
The COS method provides a way to calculating expected values (integrals) of the form
 $$v(t,x)=\int_\mathbb{R}\phi(T,y)\Gamma (t,x;T,dy),$$
and it consists of three approximation steps:
\begin{enumerate}
\item In the first step we truncate the infinite integration range to $[a,b]$ to obtain approximation $v_1$:
 $$v_1(t,x):=\int_a^b \phi(T,y)\Gamma (t,x;T,dy).$$
We assume this can be done due to the rapid decay of the distribution at infinity.
\item In the second step we replace the distribution with its cosine expansion and we get
 $$v_1(t,x):=\frac{b-a}{2}
 \sideset{}{'}\sum_{k=0}^\infty A_k(t,x;T)V_k(T),$$
 where $ \sideset{}{'}\sum$ indicates that the first term in the summation is weighted by one-half and
\begin{align}
 A_k(t,x;T)&=\frac{2}{b-a}\int_a^b\cos\left(k\pi\frac{y-a}{b-a}\right)\Gamma (t,x;T,dy),\\
 V_k(T)&=\frac{2}{b-a}\int_a^b \cos\left(k\pi\frac{y-a}{b-a}\right)\phi(T,y)dy,
\end{align}
are the Fourier-cosine series coefficients of the distribution and of the payoff function at time
$T$ respectively. Due to the rapid decay of the Fourier-cosine series coefficients, we truncate
the series summation and obtain approximation $v_2$:
 $$v_2(t,x):=\frac{b-a}{2}
 \sideset{}{'}\sum_{k=0}^{N-1} A_k(t,x;T)V_k(T).$$
\item In the third step we use the fact that the coefficients $A_k$ can be rewritten using the truncated characteristic function:
 $$A_k(t,x;T)=\frac{2}{b-a}\textnormal{Re}\left(e^{-ik\pi\frac{a}{b-a}}\int_a^be^{i\frac{k\pi}{b-a}y}\Gamma(t,x;T,dy)\right).$$
The finite integration range can be approximated as
 $$\int_a^be^{i\frac{k\pi}{b-a}y}\Gamma(t,x;T,dy)\approx\int_\mathbb{R}e^{i\frac{k\pi}{b-a}y}\Gamma(t,x;T,dy)
  =\Gh\left(t,x;T,\frac{k\pi}{b-a}\right).$$
Thus in the last step we replace $A_k$ by its approximation:
\begin{equation}\label{Fk}
 \frac{2}{b-a}\textnormal{Re}\left(
  e^{-ik\pi\frac{a}{b-a}}\Gh\left(t,x;T,\frac{k\pi}{b-a}\right)\right),
\end{equation}
and obtain approximation $v_3$:
\begin{equation}\label{approxv3}
 v_3(t,x):=\sideset{}{'}\sum_{k=0}^{N-1}\textnormal{Re}\left(
 e^{-ik\pi\frac{a}{b-a}}\Gh\left(t,x;T,\frac{k\pi}{b-a}\right)\right)V_k(T).
\end{equation}
\end{enumerate}

\medskip Next we go back to the Bermudan Put pricing problem.
Remembering that the expected value $c(t,x)$ in \eqref{bermud} can be rewritten in integral form
as in \eqref{eq:v.def1}, we have
\begin{align}
 c(t,x) = e^{-r(t_{m}-t)}\int_\Rb  v(t_m,y)\G(t,x;t_{m},dy),\qquad t\in[t_{m-1},t_{m}[.
\end{align}
Then we use the Fourier-cosine expansion \eqref{approxv3}, so that we get the approximation:
\begin{align}\label{eq:conti}
 &\hat c(t,x)= e^{-r(t_{m}-t)}\sideset{}{'}\sum_{k=0}^{N-1}
 \textnormal{Re}\left(
 e^{-ik\pi\frac{a}{b-a}}\Gh\left(t,x;t_{m},\frac{k\pi}{b-a}\right)\right)V_k(t_{m}),\qquad t\in[t_{m-1},t_{m}[\\
 &V_k(t_m)=\frac{2}{b-a}\int_a^b \cos\left(k\pi \frac{y-a}{b-a}\right)\max\{\phi(t_{m},y),c(t_{m},y)\}dy,
\end{align}
with $\phi(t,x)=\left(K-e^{x}\right)^{+}$.

Next we recover the coefficients $\left(V_k(t_m)\right)_{k=0,1,...,N-1}$ from
$\left(V_k(t_{m+1})\right)_{k=0,1,...,N-1}$. To this end, we split the integral in the definition
of $V_k(t_m)$ into two parts using the early-exercise point $x_m^*$, which is the point where the
continuation value is equal to the payoff, i.e. $c(t_m,x_m^*)=\phi(t_m,x_m^*)$; thus we have
 $$V_k(t_m)=F_k(t_{m},x_m^*)+C_k(t_{m},x_m^*),\qquad m=M-1,M-2,...,1,$$
where
\begin{equation}\label{eq:vcoef}
\begin{split}
 F_k(t_{m},x_m^*)&:=\frac{2}{b-a}\int_a^{x_m^*}\phi(t_m,y)\cos\left(k\pi\frac{y-a}{b-a}\right)dy,\\
 C_k(t_{m},x_m^*)&:=\frac{2}{b-a}\int_{x_m^*}^b c(t_m,y)\cos\left(k\pi\frac{y-a}{b-a}\right)dy,
\end{split}
\end{equation}
and $V_k(t_M) =F_k(t_{M},\log K).$
\begin{remark}
Since we have a semi-analytic formula for $\hat c(t_m,x)$, we can easily find the derivatives
with respect to $x$ and use Newton's method to find the point $x_m^*$ such that
$c(t_m,x_m^*)=\phi(t_m,x_m^*)$. A good starting point for the Newton method is $\log K$, since
$x_m^*\leq \log K$.
\end{remark}
The coefficients $F_k(t_m,x_m^*)$ can be computed analytically using $x_m^*\leq \log
K$, so that we have
\begin{align}
 F_k(t_m,x_m^*) &=\frac{2}{b-a}\int_{a}^{x_m^*}(K-e^y)\cos\left(k\pi\frac{y-a}{b-a}\right)dy\\
 &=\frac{2}{b-a}K\Psi_k(a,x_m^*)-\frac{2}{b-a}\chi_k(a,x_m^*),
\end{align}
where
\begin{align}
 \chi_k(a,x_m^*)&=\int_{a}^{x_m^*}e^y\cos\left(k\pi\frac{y-a}{b-a}\right)dy\\
 &=\frac{1}{1+\left(\frac{k\pi}{b-a}\right)^2}\left(e^{x_m^*}\cos\left(k\pi\frac{x_m^*-a}{b-a}\right)-e^{a}+\frac{k\pi e^{x_m^*}}{b-a}
 \sin\left(k\pi\frac{x_m^*-a}{b-a}\right)\right),\\
 \Psi_k(a,x_m^*)&=\int_{a}^{x_m^*}\cos\left(k\pi\frac{y-a}{b-a}\right)dy=
                \begin{cases}
                  \frac{b-a}{k\pi}\sin\left(k\pi\frac{x_m^*-a}{b-a}\right),\qquad &k\neq 0,\\
                  x_m^*-a, &k=0.
                \end{cases}
\end{align}
On the other hand, by inserting the approximation \eqref{eq:conti} for
the continuation value into the formula for $C_k(t_{m},x_{m}^*)$ have the following coefficients
$\hat C_k$ for $m =M-1,M-2,...,1$:
\begin{equation}\label{eq:contin}
 \hat C_k(t_{m},x_m^*) = \frac{2e^{-r(t_{m+1}-t_{m})}}{b-a}\sideset{}{'}\sum_{j=0}^{N-1}V_j(t_{m+1})\int_{x_m^*}^{b}
 \mathrm{Re}\left(e^{-ij\pi\frac{a}{b-a}}\Gh\left(t_{m},x;t_{m+1},\frac{j\pi}{b-a}\right)\right)
 \cos\left(k\pi\frac{x-a}{b-a}\right)dx.
\end{equation}
Thus the algorithm for pricing Bermudan options can then be summarized as follows:
\begin{figure}[h!]
\caption{Algorithm \ref{sec31}: Bermudan option valuation}
\begin{mdframed}
\begin{enumerate}
\item For $k=0,1,...,N-1$:
\begin{itemize}
\item At time $t_M$, the coefficients are exact: $V_k(t_M)=F_k(t_M,\log K)$, as in \eqref{eq:vcoef}.
\end{itemize}
\item For $m=M-1$ to 1:
\begin{itemize}
\item Determine the early-exercise point $x_m^*$ using Newton's method;
\item Compute $\hat V_k(t_m)$ using formula $\hat V_k(t_m):=F_k(t_{m},x_m^*)+\hat C_k(t_{m},x_m^*)$, \eqref{eq:vcoef} and \eqref{eq:contin}. Use an FFT for the continuation value (see Section \ref{sec32}).
\end{itemize}
\item Final step: using $\hat V_k(t_1)$ determine the option price $\hat v(0,x)=\hat c(0,x)$ using \eqref{eq:conti}.
\end{enumerate}
\end{mdframed}
\end{figure}

 \subsection{An efficient algorithm for the continuation value}\label{sec32}
In this section we derive an efficient algorithm for calculating $\hat C_k(t_{m},x_m^*)$ in
\eqref{eq:contin}. When considering an exponential L\'evy process with constant coefficients as
done in \cite{FangO09}, the continuation value can be calculated using a Fast Fourier Transform
(FFT). This can be done due to the fact that the characteristic function $\Gh(t,x;T,\x)$ can be
split into a product of a function depending only on $\xi$ and a function of the form $e^{i\xi
x}$. Note that we typically have $\xi = \frac{j\pi}{b-a}$. The integration over $x$ results in a
sum of a Hankel and Toeplitz matrix (with indices $(j+k)$ and $(j-k)$ respectively). The
matrix-vector product, with these special matrices, can be transformed into a circular convolution
which can be computed using FFTs.

From \eqref{struc3} we know that the $n$th-order approximation of the characteristic function is
of the form:
\begin{equation*}
  \Gh^{(n)}(t_m,x;t_{m+1},\x)= e^{i\xi x} \sum_{k=0}^n (x-\bar x)^k g_{n,k}(t_m,t_{m+1},\x),
\end{equation*}
where the coefficients $g_{n,k}(t,T,\x)$, with $0\le k\le n$, depend only on $t,T$ and $\x$, but not on $x$. Using \eqref{struc3} we write the continuation value as:
\begin{align}
\hat C_k(t_{m},x_m^*)
&= \sum_{h=0}^n e^{-r(t_{m+1}-t_{m})}\sideset{}{'}\sum_{j=0}^{N-1}\mathrm{Re}\left(V_j(t_m)g_{n,h}\left(t_m,t_{m+1},\frac{j\pi}{b-a}\right)M^h_{k,j}(x_m^*,b)\right),
\end{align}
where we have interchanged the sums and integral and defined:
\begin{align}M_{k,j}^h(x_m^*,b) = \frac{2}{b-a}\int_{x_m^*}^{b}  e^{ij\pi\frac{x-a}{b-a}}(x-\bar x)^h\cos\left(k\pi\frac{x-a}{b-a}\right)dx\label{eq:integraal}\end{align}
This can be written in vectorized form as:
\begin{align}
\bold{\hat C}_k(t_{m},x_m^*) =\sum_{h=1}^n e^{-r(t_{m+1}-t_m)}
\mathrm{Re}\left(\bold V(t_{m+1})\mathcal{M}^h(x_m^*,b)\Lambda^h\right),
\end{align}
where $\bold V(t_{m+1})$ is the vector $[V_0(t_{m+1}),...,V_{N-1}(t_{m+1})]^T$ and
$\mathcal{M}^h(x_m^*,b)\Lambda^h$ is a matrix-matrix product with $\mathcal{M}^h$ being a matrix
with elements $\{M_{k,j}^h\}_{k,j=0}^{N-1}$ and $\Lambda^h$ is a diagonal matrix with elements
 $$g_{n,h}\Big(t_m,t_{m+1},\frac{j\pi}{b-a}\Big),\qquad j=0,\dots,N-1.$$
We have the following theorem for calculating a generalized form of the integral in
\eqref{eq:integraal} which is used in the calculation of the continuation value.
\begin{theorem}\label{theoremhankel}
The matrix $\mathcal{M}$ with elements $\{M_{k,j}\}_{k,j=0}^{N-1}$ such that:
\begin{align}
M_{k,j}=\int e^{jx}\cos(kx)x^mdx,
\end{align}
consists of sums of Hankel and Toeplitz matrices.
\end{theorem}
\begin{proof}
Using standard trigonometric identities we can rewrite the integral as:
\begin{align}
M_{k,j} &=\int \cos(jx)\cos(kx)x^m dx+i\int \sin(jx)\cos(kx)x^m dx\\ &= M_{k,j}^H+i M_{k,j}^T,
\end{align}
where we have defined:
\begin{align}
M_{k,j}^H & = \frac{1}{2}\int \cos((j+k)x)x^mdx+\frac{1}{2}\int \sin((j+k)x)x^mdx,\\ M_{k,j}^T & =
\frac{1}{2}\int \cos((j-k)x)x^mdx+\frac{1}{2}\int \sin((j-k)x)x^mdx.
\end{align}
The following holds:
\begin{align}
 \int\cos(nx)x^mdx=&\, \frac{1}{n}x^m\sin(nx)+\sum_{i=1}^{\ceil{m/2}}(-1)^{i+1}\frac{\prod_{j=0}^{2i-2}(m-j)}{n^{2i}}\cos(nx)x^{m-(2i-1)}\\
 & -\sum_{i=1}^{\floor{m/2}}(-1)^{i+1}\frac{\prod_{j=0}^{2i-1}(m-j)}{n^{2i+1}}\sin(nx)x^{m-2i},\\
 \int \sin(nx)x^mdx=&\, -\frac{1}{n}x^m\cos(nx)+\sum_{i=1}^{\ceil{m/2}}(-1)^{i+1}\frac{\prod_{j=0}^{2i-2}(m-j)}{n^{2i}}\sin(nx)x^{m-(2i-1)}\\
 &-\sum_{i=1}^{\floor{m/2}}(-1)^{i+1}\frac{\prod_{j=0}^{2i-1}(m-j)}{n^{2i+1}}\cos(nx)x^{m-2i}.
\end{align}
It follows that $\displaystyle\{M_{k,j}^H\}_{k,j=0}^{N-1}$ is a Hankel matrix with coefficient
$(j+k)$ and  $\displaystyle\{M_{k,j}^T\}_{k,j=0}^{N-1}$  is a Toeplitz matrix with coefficient
$(j-k)$: $$\mathcal{M}_H=\begin{pmatrix}
    M_0 & M_1 & M_2 & \dots  & M_{N-1} \\
    M_1 & M_2 & \dots &   & M_N \\
    \vdots &  & & & \vdots \\
   M_{N-2} & M_{N-1} & \dots &   & M_{2N-3}\\
    M_{N-1} & \dots &  & M_{2N-3}  & M_{2N-2}
\end{pmatrix},$$
$$\mathcal{M}_T=\begin{pmatrix}
    M_0 & M_1 & \dots & M_{N-2}  & M_{N-1} \\
    M_{-1} & M_0 & M_1& \dots  & M_{N-2} \\
    \vdots &  & \ddots& & \vdots \\
   M_{2-N} & \dots & M_{-1} &M_0   & M_1\\
    M_{1-N} & M_{2-N}&  & M_{-1}  & M_0
\end{pmatrix},$$
where we have defined
\begin{align}
M_j = \frac{1}{2}\int \cos(jx)x^mdx+\frac{1}{2}\int \sin(jx)x^mdx.
\end{align}
\end{proof}
From Theorem \ref{theoremhankel} we see that $\mathcal{M}^h(x_m^*,b)$ with elements $M_{k,j}^h$ consists of a sum of a Hankel and Toeplitz matrix.
\begin{example}
We derive explicitly the Hankel and Toeplitz matrices for $m=0$ and $m=1$. We calculate the
indefinite integral
\begin{align}
M_{k,j} = \frac{2}{b-a}\int e^{ij\pi\frac{x-a}{b-a}}\cos\left(k\pi\frac{x-a}{b-a}\right)(x-\bar
x)^mdx.
\end{align}
Suppose $m=0$, in this case we have $M_{k,j} = M_{k,j}^{H}+M_{k,j}^T$, with:
\begin{align}
M_{k,j}^H &= -\frac{i\exp\left(i\frac{(j+k)\pi(x-a)}{b-a}\right)}{\pi(j+k)},\\ M_{k,j}^T &=
-\frac{i\exp\left(i\frac{(j-k)\pi(x-a)}{b-a}\right)}{\pi(j-k)},
\end{align}
where $\displaystyle\{M_{k,j}^H\}_{k,j=0}^{N-1}$ is a Hankel matrix and
$\displaystyle\{M_{k,j}^T\}_{k,j=0}^{N-1}$ is a Toeplitz matrix with
\begin{align}
 M_{j} =
 \begin{cases}
  \frac{x}{b-a},\qquad &j=0,\\
  \frac{i\exp\left(i\frac{j\pi(x-a)}{b-a}\right)}{\pi j},\quad &j\neq 0.
 \end{cases}
\end{align}
Suppose $m=1$, in this case we have:
\begin{align}
M_{k,j}^H &= -\frac{a-b}{(j-k)^2\pi^2}\exp\left(i(j-k)\pi\frac{(x-a)}{b-a}\right)-\frac{x-\bar
x}{(j-k)\pi}i\exp\left(i(j-k)\pi\frac{(x-a)}{b-a}\right),\\ M_{k,j}^T &=
-\frac{a-b}{(j+k)^2\pi^2}\exp\left(i(j+k)\pi\frac{(x-a)}{b-a}\right)-\frac{x-\bar
x}{(j+k)\pi}i\exp\left(i(j+k)\pi\frac{(x-a)}{b-a}\right),
\end{align}
where $\displaystyle\{M_{k,j}^H\}_{k,j=0}^{N-1}$ is a Hankel matrix and
$\displaystyle\{M_{k,j}^T\}_{k,j=0}^{N-1}$ is a Toeplitz matrix, with
\begin{align}
 M_{j} =
 \begin{cases}
  \frac{x(x-\bar x)}{b-a},\qquad &j=0,\\
  -\frac{a-b}{j^2\pi^2}\exp\left(ij\pi\frac{(x-a)}{b-a}\right)-\frac{x-\bar
  x}{j\pi}i\exp\left(ij\pi\frac{(x-a)}{b-a}\right),\quad &j\neq 0.
 \end{cases}
\end{align}
\end{example}
\begin{remark}
If we take $\bar x=x$, which is most common in practice, the formulas are simplified significantly and only
the case of $m=0$ is relevant. In this case the characteristic function is simply $e^{i\xi
x}$ times a sum of terms depending only on $t_m$, $t_{m+1}$ and $\xi=\frac{j\pi}{b-a}$:\begin{equation*}
  \Gh^{(n)}(t_m,x;t_{m+1},\x)= e^{i\xi x}  g_{n,0}(t_m,t_{m+1},\x).
\end{equation*}
\end{remark}
Using the split into sums of Hankel and Toeplitz matrices we can write the continuation value in matrix form as:
\begin{align}
 \boldsymbol{\hat{C}} (t_{m},x_m^*) = \sum_{h=0}^n e^{-r(t_{m+1}-t_m)}\mathrm{Re}\left((\mathcal{M}^h_H+\mathcal{M}^h_T)\boldsymbol{u}^l\right),
\end{align}
where $\mathcal{M}^h_H=\{M^{H,h}_{k,j}(x_m^*,b)\}_{k,j=0}^{N-1}$ is a Hankel matrix and
$\mathcal{M}^l_T=\{M^{T,h}_{k,j}(x_m^*,b)\}_{k,j=0}^{N-1}$ is a Toeplitz matrix and
$\boldsymbol{u}^h=\{u_j^h\}_{j=0}^{N-1}$, with
$u_j^h=g_{n,h}\left(t_m,t_{m+1},\frac{j\pi}{b-a}\right)V_j(t_{m+1})$ and
$u_0^h=\frac{1}{2}g_{n,h}\left(t_m,t_{m+1},0\right)V_0(t_{m+1})$.

\bigskip We recall that the circular convolution, denoted by $\circledast$, of two vectors is equal
to the inverse discrete Fourier transform $(\mathcal{D}^{-1})$ of the products of the forward
DFTs, $\mathcal{D}$, i.e.:
 $$\bold{x} \circledast \bold{y} = \mathcal{D}^{-1}\{\mathcal{D}(\bold{x})\cdot\mathcal{D}(\bold{y})\}.$$
For Hankel and Toeplitz matrices we have the following result:
\begin{theorem}
For a Toeplitz matrix $\mathcal{M}_T$, the product $\mathcal{M}_T \bold{u}$
is equal to the first $N$ elements of $\bold{m}_T \circledast \bold{u}_T $, where $\bold{m}_T$ and $\bold{u}_T$ are $2N$ vectors defined by
\begin{align}
&\bold{m}_T = [M_0, M_{-1},M_{-2},...,M_{1-N},0,M_{N-1},M_{N-2},...,M_1]^T,\\
&\bold{u}_T = [u_0,u_1,... ,u_{N-1},0,...,0]^T.
\end{align}
 For a Hankel matrix $\mathcal{M}_H$, the product $\mathcal{M}_H \bold{u}$ is equal to the first $N$ elements of $\bold{m_H} \circledast \bold{u_H} $ in reversed order, where $\bold{m}_H$ and $\bold{u}_H$ are $2N$ vectors defined by
 \begin{align}
 &\bold{m}_H = [M_{2N-1},M_{2N-2},... ,M_1,M_0]^T\\
 &\bold{u}_H = [0,...,0,u_0,u_1,...,u_{N-1}]^T.
 \end{align}
\end{theorem}
Summarizing, we can calculate the continuation value $\bold{\hat C}(t_m,x_m^*)$ using the
algorithm in Figure \ref{fig2}.
\begin{figure}[h!]\label{fig2}
\caption{Algorithm \ref{sec32}: Computation of $\bold{\hat C}(t_m,x_m^*)$}
\begin{mdframed}
\begin{enumerate}
\item For $h=0,...,n$:
\begin{itemize}
\item Compute $M^h_j(x_1,x_2)$
\item Construct $\bold{m}^h_H$ and $\bold{m}^h_T$
\item Compute $\boldsymbol{u}^h(t_m)=\{u_j^h\}_{j=0}^{N-1}$
\item Construct $\bold{u}^h_T$ by padding $N$ zeros to $\boldsymbol{u}^h(t_m)$
\item $\bold{MTu}^h=$ the first $N$ elements of $\mathcal{D}^{-1}\{\mathcal{D}(\bold{m}_T^h)\cdot\mathcal{D}(\bold{u}_T^h)\}$
\item $\bold{MHu}^h=$ reverse$\{$the first $N$ elements of $\mathcal{D}^{-1}\{\mathcal{D}(\bold{m}_H^h)\cdot\bold{sgn}\cdot \mathcal{D}(\bold{u}_T^h)\}\}$
\end{itemize}
\item Compute the continuation value using $\bold{\hat C}(t_m,x_m^*)= \sum\limits_{h=0}^n e^{-r(t_{m+1}-t_m)}\mathrm{Re}(\bold{MTu}^h+\bold{MHu}^h)$.
\end{enumerate}
\end{mdframed}
\end{figure}
\\
The continuation value requires five DFTs  for each $h=0,...,n$, and a DFT is calculated using the
FFT. In practice it is most common to have $\bar x = x$ and in this case we only need five FFTs.
The computation of $F_k (t_m,x_m^*)$ is linear in $N$. The overall complexity of the  method is
dominated by the computation of $\hat C(t_m,x_m^*)$, whose complexity is $O(N \log_2 N)$ with the
FFT. The complexity of the calculation for option value at time 0 is $O(N)$. If we have a Bermudan
option with $M$ exercise dates, the overall complexity will be $O((M-1)N\log_2N)$.

\begin{remark}[{\bf American options}] The prices of American options can be obtained by
applying a Richardson extrapolation (see, for instance, \cite{american}) on the prices of a few
Bermudan options with a small number of exercise dates. Let $v_{M}$ denote the value of a Bermudan
option with maturity $T$ and a number $M$ of early exercise dates that are $ \tfrac{T}{M}$ years
apart. Then, for any $d\in \mathbb{N}$, the following 4-point Richardson extrapolation scheme
 $$
 \frac{1}{21}\left(64v_{2^{d+3}}-56v_{2^{d+2}}+14v_{2^{d+1}}-v_{2^d}\right)$$
gives an approximation of the corresponding American option price.
\end{remark}

\begin{remark}[{\bf The Greeks}]\label{remarkgreeks} The approximation method can also be used to
calculate the Greeks at almost no additional cost. In the case of $\bar x = x$, we have the
following approximation formulas for Delta and Gamma:
\begin{align}
 \hat \Delta  =&\ e^{-r(t_{1}-t_0)}\sideset{}{'}\sum_{k=0}^{N-1}\textnormal{Re}\left(e^{ik\pi\frac{x-a}{b-a}}\left(\frac{ik\pi}{b-a}g_{n,0}
 \left(t_0,t_{1},\frac{k\pi}{b-a}\right)+g_{n,1}\left(t_0,t_{1},\frac{k\pi}{b-a}\right)\right)\right)\hat V_k(t_1),\\
 \hat \Gamma  =&\ e^{-r(t_{1}-t_0)}\sideset{}{'}\sum_{k=0}^{N-1}\textnormal{Re}\bigg(e^{ik\pi\frac{x-a}{b-a}}\bigg(-\frac{ik\pi}{b-a}
 g_{n,0}\left(t_0,t_{1},\frac{k\pi}{b-a}\right)-g_{n,1}\left(t_0,t_{1},\frac{k\pi}{b-a}\right)\\
 &+2\frac{ik\pi}{b-a}g_{n,1}\left(t_0,t_{1},\frac{k\pi}{b-a}\right)
 +\left(\frac{ik\pi}{b-a}\right)^2g_{n,0}\left(t_0,t_{1},\frac{k\pi}{b-a}\right)+2g_{n,2}\left(t_0,t_{1},\frac{k\pi}{b-a}\right)\bigg)\bigg)\hat V_k(t_1).
\end{align}
\end{remark}

\section{Error estimates} \label{section3}
The error in our approximation consists of the error of the COS method and the error in the
adjoint expansion of the characteristic function. The error of the COS method depends on the
truncation of the integration range $[a,b]$ and the truncation of the infinite summation of the
Fourier-cosine expansion by $N$. The density rapidly decays to zero as $y\rightarrow\pm\infty$. Then the overall error can be bounded as follows:
 $$\epsilon_1(x;N,[a,b])\leq Q\left|\int_{\mathbb{R}\backslash[a,b]}\Gamma(t,x;T,dy)\right|+\left|\frac{P}{(N-1)^{\beta-1}}\right|,$$
where $P$ and $Q$ are constants not depending on $N$ or $[a,b]$ and $\beta\geq n\geq 1$, with $n$
being the algebraic index of convergence of the cosine series coefficients. For a sufficiently
large integration interval $[a,b]$, the overall error is dominated by the series truncation error,
which converges exponentially. The error in the backward propagation of the coefficients
$V_k(t_m)$ is defined as $\epsilon_2(k,t_m):=V_k(t_m)-\hat V_k(t_m)$. With $[a,b]$ sufficiently
large and a probability density function in $C^\infty([a,b])$, the error $\epsilon_1(k,t_m)$ converges exponentially in $N$. For a
detailed derivation on the error of the COS method see \cite{FangO08} and \cite{FangO09}.

We now present the error estimates for the adjoint expansion of the characteristic function at
orders zero and one. We consider for simplicity a model with time-independent coefficients
\begin{equation}\label{eq:modelerrorest}
 X_t = x + \int_0^t\mu(X_s)ds+\int_0^t\sigma(X_s)dW_s+\int_0^t\int_\mathbb{R}\eta(X_{s-})zd\tilde
 N(s,dz),
\end{equation}
where we have defined as usual $d\tilde N(t,dz) = dN(t,dz) - \nu(dz)dt$. This model is similar to
the model we considered initially in \eqref{eq:hetmodel}; only now we deal with slightly
simplified version and assume that the dependency on $X_t$ in the measure can be factored out,
which is often enough the case.

Let $\tilde X_t$ be the 0th-order approximation of the model in \eqref{eq:modelerrorest} {with
$\bar x = x$}, that is
\begin{align}\label{eq:Xtilda}
 \tilde X_t = x +\int_0^t\mu(x)ds+ \int_0^t\sigma(x)dW_s+\int_0^t\int_\mathbb{R}\eta(x)zd\tilde
 N(s,dz).
\end{align}
The characteristic exponent of $\tilde X_t-x$ is
\begin{equation}\label{charex}
 \psi(\xi) = i\xi\mu(x)-\frac{\sigma(x)^2}{2}\x^2-\eta(x)\int_\mathbb{R}\nu(dz)(e^z-1-z)i\xi+\eta(x)\int_\mathbb{R}\nu(dz)(e^{iz\xi}-1-iz\xi).
\end{equation}

\begin{theorem}\label{theoremzero}
Let $n=0,1$ and assume that the coefficients $\mu,\s,\y$ are continuously differentiable with
bounded derivatives up to order $n$. Let $\Gh^{(n)}(0,x;t,\x)$ in \eqref{adapprox} be the
$n$th-order approximation of the characteristic function. Then, for any $T>0$ there exists a
positive constant $C$ that depends only on $T$, on the norms of the coefficients and on the L\'evy
measure $\nu$, such that
\begin{align}\label{asa}
 \left| \Gh(0,x;t,\x)-\Gh^{(n)}(0,x;t,\x)\right|&\leq C\left(1+|\x|^{1+3n}\right)t^{n+1},\qquad t\in[0,T],\ \x\in\R.
\end{align}
\end{theorem}
\begin{proof}
For the proof we refer to Appendix \ref{proof1}.
\end{proof}
\begin{remark}
The proof of Theorem \ref{theoremzero} can be generalized to obtain error bounds for any $n\in\N$:
however, one can see that, for $n\ge2$, the order of convergence improves only in the diffusive
part, according to the results proved in \cite{LPP4}.
\end{remark}

\section{Numerical tests} \label{section4}
For the numerical examples we use the second-order approximation of the characteristic function.
We have found this to be sufficiently accurate by numerical experiments and theoretical error
estimates. The formulas for the second-order approximation are simple, making the method easy to
implement. For the COS method, unless otherwise mentioned, we use $N=200$ and $L = 10$, where $L$
is the parameter used to define the truncation range $[a,b]$ as follows:
$$[a,b]:=\left[c_1-L\sqrt{c_2+\sqrt{c_4}},c_1+L\sqrt{c_2+\sqrt{c_4}}\right],$$ where $c_n$ is the
$n$th cumulant of log-price process $\log S$, as proposed in \cite{FangO08}. The cumulants are
calculated using the 0th-order approximation of the characteristic function. A larger $N$ and $L$
has little effect on the price, since a fast convergence is achieved already for small $N$ and
$L$. We compare the approximated values to a 95\% confidence interval computed with a
Longstaff-Schwartz method with $10^5$ simulations and $250$ time steps per year. Furthermore, in
the expansion we always use $\bar x = x$.
\subsection{Tests under CEV-Merton dynamics}
Consider a process under the CEV-Merton dynamics:
 $$dX_t = \left(r - a(x) - \lambda
\left(e^{m+\delta ^2/2}-1\right)\right)dt+\sqrt{2 a(x)}dW_t + \int_\mathbb{R}d\tilde N_t(t,dz)z,$$
with
\begin{align}
&a(x) = \frac{\sigma_0^2 e^{2(\beta-1)x}}{2},\\
&\nu(dz) = \lambda\frac{1}{\sqrt{2\pi\delta^2}}\exp\left(\frac{-(z-m)^2}{2\delta^2}\right)dz,\\
&\psi(\xi) =
-a_0(\xi^2+i\xi)+ir\xi- i\lambda \left(e^{m+\delta^2/2}-1\right)\xi +\lambda\left(e^{mi\xi-\delta^2\xi^2/2}-1\right).
\end{align}
We use the following parameters $S_0=1$, $r=5\%$, $\sigma_0=20\%$, $\beta=0.5$, $\lambda = 30\%$,
$m=-10\%$, $\delta=40\%$ and compute the European and Bermudan option values.
\begin{table}[h!]  \caption{Prices for a European and a Bermudan Put option (expiry $T=0.25$ with 3 exercise dates, expiry $T=1$ with 10 exercise dates
and expiry $T=2$ with 20 exercise dates) in the CEV-Merton model for the 2nd-order approximation
of the characteristic function, and a Monte Carlo method.}
\begin{center}
\begin{tabular}{ l|l|l|l|l|l }
\hline & &\multicolumn{2}{ |c }{European} &\multicolumn{2}{ |c }{Bermudan} \\ \hline \hline T&K &
MC 95\% c.i. & Value &MC 95\% c.i.&Value \\ \hline
0.25&0.6&0.001240-0.001433&0.001326&0.001243-0.001431&0.001307\\
&0.8&0.005218-0.005679&0.005493&0.005314-0.005774&0.005421\\
&1&0.04222-0.04321&0.04275&0.04274-0.04371&0.04304\\
&1.2&0.1923-0.1938&0.1935&0.1979-0.1989&0.1981\\ &1.4&0.3856-0.3872&0.3866&0.3948-0.3958&0.3955\\
&1.6&0.5812-0.5829&0.5825&0.5940-0.5950&0.5941\\ \hline \hline
1&0.6&0.006136-0.006573&0.006579&0.006307-0.006729&0.006096\\
&0.8&0.02526-0.02622&0.02581&0.02617-0.02711& 0.02520\\ &1&0.08225-0.08395 &
0.08250&0.08480-0.08640& 0.08593\\ &1.2 & 0.1965-0.1989&0.1977&0.2097-0.2115&0.2132\\
&1.4&0.3560-0.3589 & 0.3574 & 0.3946-0.3957& 0.3954
\\
&1.6 & 0.5341-0.5385 & 0.5364 & 0.5930-0.5941& 0.5932 \\ \hline \hline 2&
0.6&0.01444-0.01513&0.01529&0.01528-0.01594&0.01365\\
&0.8&0.04522-0.04655&0.04613&0.04596-0.04719&0.04659\\
&1&0.1046-0.1067&0.1077&0.1149-0.1168&0.1171\\ &1.2&0.2054-0.2083&0.2065&0.2319-0.2341&0.2345\\
&1.4&0.3351-0.3386&0.3382&0.3968-0.3987&0.3991\\ &1.6&0.4904-0.4944&0.4919&0.5927-0.5938&0.5935\\
\hline \hline
\end{tabular}
\label{tabmer2}
\end{center}
\end{table}\\
We present the results in Table \ref{tabmer2}. The option value for both the Bermudan options as
well as the European options appears to be accurate. Since the COS method has a very quick
convergence, already for $N=64$ the error becomes stable. For at-the-money strikes we have
$\log_{10}|\textnormal{error}|\approx 3.5$. The use of the second-order approximation of the
characteristic function is justified by the fact that the option value (and thus the error)
stabilizes starting from the second-order approximation. Furthermore, it is noteworthy that the
0th-order approximation is already very accurate.
\\\\
The computer used in the experiments has an Intel Core i7 CPU with a 2.2 GHz processor. The CPU
time of the calculations depends on the number of exercise dates. Assuming we use the second-order
approximation of the characteristic function, if we have $M$ exercise dates the CPU time will be
$5\cdot M$ ms.
\begin{remark}
The method can be extended to include time-dependent coefficients.
The accuracy and speed of the method will be of the same order as for time-independent coefficients.
\end{remark}
\begin{remark}
The Greeks can be calculated at almost no additional cost using the formulas presented in
\ref{remarkgreeks}. Numerically, the order of convergence is algebraic and is the same for both the exact
characteristic function as for the 2nd-order approximation.
\end{remark}

\subsection{Tests under the CEV-Variance-Gamma dynamics}
Consider the jump process to be a Variance-Gamma process. The VG process, is obtained by replacing
the time in a Brownian motion with drift $\theta$ and standard deviation $\varrho$, by a Gamma
process with variance $\kappa$ and unitary mean. The model parameters $\varrho$ and $\kappa$ allow
to control the skewness and the kurtosis of the distribution of stock price returns. The VG
density is characterized by a fat tail and is thus used as a model in situations where small and
large asset values are more probable than would be the case for the lognormal distribution. The
L\'evy measure in this case is given by:
 $$\nu(dx)=\frac{e^{-\lambda_1x}}{\kappa x}\caratt_{\{x>0\}}dx+\frac{e^{\lambda_2x}}{\kappa |x|}\caratt_{\{x<0\}}dx,$$
where
 $$\lambda_1 = \left(\sqrt{\frac{\theta^2\kappa^2}{4}+\frac{\varrho^2\kappa}{2}}+\frac{\theta\kappa}{2}\right)^{-1},\;\;\;\;\;\lambda_2
 =\left(\sqrt{\frac{\theta^2\kappa^2}{4}+\frac{\varrho^2\kappa}{2}}-\frac{\theta\kappa}{2}\right)^{-1}.$$
Furthermore we have
\begin{align}
&a(x) = \frac{\sigma_0^2 e^{2(\beta-1)x}}{2},\\
&\mu(t,x)= r+\frac{1}{\kappa}\log\left(1-\kappa\theta-\frac{\kappa\varrho^2}{2}\right) - a(x), \\
& \psi(\xi) = -a_0(\xi^2+i\xi)+ir\xi+i\frac{1}{\kappa}\log\left(1-\kappa\theta-\frac{\kappa\varrho^2}{2}\right)\xi
 -\frac{1}{\kappa}\log\left(1-i\kappa\theta\xi+\frac{\xi^2\kappa\varrho^2}{2}\right).
\end{align}
We use the following parameters
$S_0=1$, $r=5\%$, $\sigma_0=20\%$, $\beta=0.5$, $\kappa = 1$, $\theta = -50\%$, $\varrho = 20\%$. The results for the European and Bermudan option are presented in Table \ref{tabgamma1}.
\begin{table}[h!]
   \caption{Prices for a European and a Bermudan Put option (10 exercise dates, expiry $T=1$)
   in the CEV-VG model for the 2nd-order approximation of the characteristic function, and a Monte Carlo method.}
\begin{center}
\begin{tabular}{ l|l|l|l|l }
\hline &\multicolumn{2}{ |c }{European} &\multicolumn{2}{ |c }{Bermudan} \\ \hline \hline K & MC
95\% c.i. & Value &MC 95\% c.i.&Value \\ \hline 0.6&0.03090-0.03732&0.03546&0.03756-0.03876&0.03749\\
0.8&0.08046-0.08247&0.08029&0.08290-0.08484&0.08395\\
1&0.1507-0.1531&0.1511&0.1572-0.1600&0.1594\\
1.2&0.2501-0.2538&0.2522&0.2634-0.2668&0.2685\\
1.4&0.3831-0.3876&0.3847&0.4073-0.4108&0.4137\\
1.6&0.5430-0.5479&0.5436&0.5920-0.5938&0.5937\\ \hline \hline
\end{tabular}
\label{tabgamma1}
\end{center}
\end{table}

\subsection{CEV-like L\'evy process with a state-dependent measure and default}
In this section we consider a model similar to the one used in \cite{JacquierLorig2013}. The model is
defined with local volatility, local default and a state-dependent L\'evy measure as follows:
\begin{align}
&a(x)=\frac{1}{2}(b_0^2+\epsilon_1 b_1^2\eta(x)),\nonumber\\
&\gamma(x)=c_0+\epsilon_2 c_1\eta(x),\nonumber\\
&\nu(x,dz)=\epsilon_3\nu_N(dz)+\epsilon_4 \eta(x)\nu_N(dz),\nonumber \\
&\eta(x)=e^{\beta x}.\label{eq:modellevy}
\end{align}
We will consider Gaussian jumps, meaning that
\begin{align}
 \nu_N(dz) = \lambda\frac{1}{\sqrt{2\pi\delta^2}}\exp\left(\frac{-(z-m)^2}{2\delta^2}\right)dz.\label{eq:gaussjump}
\end{align}
\noindent
The regular CEV model has several shortcomings: the volatility for instance drops to zero as the
underlying approaches infinity; also the model does not allow the underlying to experience jumps.
This model tries to overcome these shortcomings, while still retaining CEV-like behaviour through
$\eta(x)$. The local volatility function $\sigma(x)$ behaves asymptotically like the CEV model,
$\sigma(x)\sim \sqrt{\epsilon_1}b_1e^{\beta x/2}$ as $x\rightarrow -\infty$, reflecting the fact
that the volatility tends to increase as the asset price drops (the leverage effect). Jumps of
size $dz$ arrive with a state-dependent intensity of $\nu(x,dz)$. Lastly, a default arrives with
intensity $\gamma(x)$. The default function $\gamma(x)$ behaves asymptotically like $\epsilon_2 c_1
e^{\beta x}$ as $x\rightarrow -\infty$, reflecting the fact that a default is more likely to occur
when the price goes down.
\\
In Table \ref{table3} the results are presented for a model as defined in \eqref{eq:modellevy}
without default, meaning that $c_0=c_1=0$ and with a state-dependent jump measure, so $\nu(x,dz)=\eta(x)\nu_N(dz)$. In this case we have
$$\psi(\xi)=ir\xi-a_0(\xi^2-i\xi)-\lambda\nu_0(e^{m+\delta^2/2}-1)i\xi+\lambda\nu_0(e^{mi\xi-\delta^2\xi^2/2}-1),$$
where $a_0=\frac{1}{2}b_1^2e^{\beta \bar x}$ and $\nu_0(dz)=e^{\beta \bar x}\nu_N(dz)$. The other parameters are chosen as: $b_1 = 0.15$, $b_0=0$, $\beta = -2$, $\lambda=20\%$,
$\delta=20\%$, $m=-0.2$, $S_0=1$, $r=5\%$, $\epsilon_1 = 1$, $\epsilon_3=0$, $\epsilon_4=1$, the number of exercise dates is 10 and
$T=1$.
\begin{table}[h!]
\caption{Prices for a European and a Bermudan Put option (10 exercise dates, expiry $T=1$) in the
CEV-like model with state-dependent measure for the 2nd-order approximation characteristic
function, and a Monte Carlo method.}
\begin{center}
\begin{tabular}{ l|l|l|l|l }
\hline &\multicolumn{2}{ |c }{European} &\multicolumn{2}{ |c }{Bermudan} \\ \hline \hline K & MC
95\% c.i. & Value &MC 95\% c.i.&Value \\ \hline 0.8&0.01025-0.01086&0.009385&0.01068-0.01125&0.01024\\
 1&0.04625-0.04745&0.04817& 0.05141-0.05253&0.05488\\
1.2&0.1563-0.1582&0.1564&0.1942-0.1952&0.1952\\
1.4&0.3313-0.3334&0.3314&0.3927-0.3934&0.3930\\
1.6&0.5207-0.5229&0.5218&0.5919-0.5926&0.5920\\
1.8&0.7103-0.7124&0.7122&0.7906-0.7913&0.7910\\ \hline\hline
\end{tabular}
\end{center}\label{table3}
\end{table}
From the results for both the European option and the Bermudan option we see that the method
performs very accurately, even for deeply in-the-money strikes.
\\\\
In Table \ref{table4} the results are presented for the value of a defaultable Put option. In case
of default prior to exercise the Put option payoff is 0, in case of no default the value is
$(K-S_t)^+$, depending on the exercise time. We look at the model as defined in
\eqref{eq:modellevy} with the possibility of default and consider state-independent jumps, meaning that we have $\gamma(x) = \eta(x)$ and $\nu(x,dz)=\nu_N(dz)$. We have $$\psi(\xi)=ir\xi-a_0(\xi^2-i\xi)+\gamma_0i\xi-\gamma_0-\lambda(e^{m+\delta^2/2}-1)i\xi+\lambda(e^{mi\xi-\delta^2\xi^2/2}-1),$$ where
$a_0=\frac{1}{2}b_1^2e^{\beta \bar x}$ and $\gamma_0=c_1e^{\beta \bar x}$. The other parameters are $b_0=0$, $b_1 = 0.15$, $\beta = -2$,
$c_0=0$, $c_1=0.1$, $S_0=1$, $r=5\%$, $\epsilon_1 = 1$, $\epsilon_2=1$, $\epsilon_3=1$, $\epsilon_4=0$, the number of exercise dates is 10 and
$T=1$.
\begin{table}[h!]
\caption{Prices for a European and a Bermudan Put option (10 exercise dates, expiry $T=1$) in the
CEV-like model with default for the 2nd-order approximation characteristic function, and a Monte
Carlo method.}
\begin{center}
\begin{tabular}{ l|l|l|l|l }
\hline &\multicolumn{2}{ |c }{European} &\multicolumn{2}{ |c }{Bermudan} \\ \hline \hline K & MC
95\% c.i. & Value &MC 95\% c.i.&Value \\ \hline
0.8&0.002905-0.003175&0.003061&0.005876-0.006245&0.006361\\
1&0.01845-0.01918&0.01893&0.03419-0.03506&0.03520\\
1.2&0.08148-0.08296&0.08297&0.1820-0.1827&0.1824\\
1.4&0.2184-0.2205&0.2173&0.3793-0.3801&0.3792\\
1.6&0.3867-0.3892&0.3841&0.5752-0.5763&0.5763\\
1.8&0.5597-0.5638&0.5556&0.7727-0.7739&0.7733\\
\hline\hline
\end{tabular}
\end{center}\label{table4}
\end{table}

\newpage
\appendix
\section{Proof of Theorem \ref{theoremzero}}\label{proof1}
Let $X$ and $\tilde X$ be as in \eqref{eq:modelerrorest} and \eqref{eq:Xtilda} respectively. We
first  prove that
\begin{equation}\label{esti1}
  E[|X_t-\tilde X_t|^{2}]\le C\left(\kappa_2 t^2+\kappa_1^2 t^3\right),\qquad t\in[0,T],
\end{equation}
for some positive constant $C$ that depends only on $T$, on the Lipschitz constants of the
coefficients $\mu$, $\sigma$, $\eta$ and on the L\'evy measure $\nu$. Here $\kappa_1=-\psi'(0)$
and $\kappa_2=-\psi''(0)$ where $\psi$ in \eqref{charex} is the characteristic exponent of the
L\'evy process $(\tilde X_t-x)$.

Using the H\"older inequality, the It\^o isometry (see, for instance, \cite{Pascucci2011}) 
and the Lipschitz continuity 
of $\eta$, $\mu$ and $\sigma$, the mean squared error is bounded by:
\begin{align}
 E\left[|X_t-\tilde X_t|^2\right] \leq&\ 3E\left[\left(\int_0^t(\mu(X_s)-\mu(x))ds\right)^2\right]+3E\left[\left(\int_0^t(\sigma(X_s)-\sigma(x))dW_s\right)^2\right]\\
 &+3E\left[\left(\int_0^t\int_\mathbb{R}(\eta(X_{s-})-\eta(x))zd\tilde N(s,dz)\right)^2\right]\\
 \label{cum23}
 \leq&\ C\int_0^tE\left[|\tilde X_s-x|^2\right]ds+C\int_0^tE\left[|X_s-\tilde X_s|^2\right]ds,
\end{align}
where
\begin{align}
 C=6\left(\left\|\mu'\right\|_\infty^2+\left\|\sigma'\right\|_\infty^2+\left\|\eta'\right\|_\infty^2\int_\mathbb{R}z^2\nu(dz)\right).
\end{align}
Now we recall the following relationship between the first and second moment and cumulants
\begin{align}
 E[(\tilde X_s-x)]=c_1(s),\qquad
 E[(\tilde X_s-x)^2]=c_2(s)+c_1(s)^2,
\end{align}
where
  $$c_n(s)=\frac{s}{i^n}\frac{\partial^n \psi(\xi)}{\partial\xi^n}\bigg|_{\xi=0},$$
and $\psi(\xi)$ is the characteristic exponent of $(\tilde X_s-x)$. Thus we have
\begin{align}\label{cum12}
 E\left[|\tilde X_{s}-x|^2\right]= 
 \kappa_2s+\kappa_1^2s^2.
\end{align}
Plugging \eqref{cum12} into \eqref{cum23} we get
  $$E[|X_t-\tilde X_t|^2]\le C\left(\frac{\kappa_2}{2}t^2+\frac{\kappa_1^2}{3}t^3\right)+C \int_0^tE\left[|X_{s}-\tilde X_{s}|^2\right]ds,$$
and therefore estimate \eqref{esti1} follows by applying the Gronwall inequality in the form
\begin{align}
 \phi(t)\leq \alpha(t)+C\int_0^t\phi(s)ds\ \implies\ \phi(t)\leq
 \alpha(t)+C\int_0^t\a(s)e^{C(t-s)}ds,
\end{align}
that is valid for any $C\ge0$ and $\phi$, $\a$ continuous functions.

From \eqref{esti1} and \eqref{cum12} we can also deduce that
\begin{equation}\label{eses}
 E\left[\left|X_{t}-x\right|^{2}\right]\le
 2E\left[\big|X_{t}-\tilde{X}_{t}\big|^{2}\right]+2E\left[\big|\tilde{X}_{t}-x\big|^{2}\right]\leq
 C\left(\kappa_2t+\kappa_1^2t^2\right),\qquad t\in[0,T].
\end{equation}
Moreover, from \eqref{esti1} we also get the following error estimate for the expectation of a
Lipschitz payoff function $v$:
\begin{equation}\label{eses1}
 \left|E\left[v(X_{t})\right]-E[v(\tilde{X}_{t})]\right|\le C\sqrt{\kappa_2t+\kappa_1^2t^2},\qquad t\in[0,T],
\end{equation}
where now $C$ also depends on the Lipschitz constant of $v$. In particular, taking
$v(x)=e^{ix\x}$, this proves \eqref{asa} for $n=0$.

\medskip Next we prove \eqref{asa} for $n=1$.

Proceeding as in the proof of Lemma 6.23 in \cite{LPP4} with $u(0,x)=\Gh(0,x;t,\x)$ and $\xbar=x$,
we find
\begin{align}
 \Gh(0,x;t,\x)-\Gh^{(1)}(0,x;t,\x)&=\int_{0}^{t}E\left[(L-L_0)\hat G^1(s,X_s;t,\x)+(L-L_1)\hat G^0(s,X_s;t,\x)\right]ds,
\end{align}
where the 1st-order approximation is as usual
\begin{align}
\Gh^{(1)}(s,X;t,\x)=\hat G^{0}(s,X;t,\x)+\hat G^{1}(s,X;t,\x),
\end{align}
with
\begin{align}
 &\hat G^{0}(s,X;t,\x)= e^{i X \x+(t-s)\psi(\x)},\\ 
 &\hat G^{1}(s,X;t,\x)= e^{i X \x+(t-s)\psi(\x)}g_0^{(1)}(t-s,\x),
\end{align}
and $g_0^{(1)}$ as in \eqref{g01}.
Using the Lagrangian remainder of the Taylor expansion, we have
\begin{align}
 L-L_0&=\gamma'(\e')(X-x)(\partial_X-1)+a'(\e')(X-x)(\partial_{XX}-\partial_{X})+\eta'(\e')(X-x)\int_\mathbb{R}\nu(dz)(e^z-1-z)\partial_{X}\\
 &+\eta'(\e')(X-x)\int_\mathbb{R}\nu(dz)(e^{z\partial_{X}}-1-z\partial_{X}),\\
 L-L_1&=\frac{1}{2}\gamma''(\e'')(X-x)^2(\partial_X-1)+\frac{1}{2}a''(\e'')(X-x)^2(\partial_{XX}-\partial_{X})\\
 &+\frac{1}{2}\eta''(\e'')(X-x)^2\int_\mathbb{R}\nu(dz)(e^z-1-z)\partial_{X}+\frac{1}{2}\eta''(\e'')(X-x)^2\int_\mathbb{R}\nu(dz)(e^{z\partial_{X}}-1-z\partial_{X}),
\end{align}
for some $\e',\e'' \in[x,X]$. Now, $|\hat G^{0}|\le 1$ because $\hat G^{0}$ is the characteristic
function of the process $\tilde{X}$ in \eqref{eq:Xtilda}; thus, we have
\begin{equation}\label{sss}
  \left|(L-L_1)\hat G^0(s,X_s;t,\x)\right|\le C(1+|\x|^{2}) \left|X_s-x\right|^{2}.
\end{equation}
On the other hand, from \eqref{g01} we have
  $$\left|g_0^{(1)}(t-s,\x)\right|\le C(t-s)^{2}\left(1+|\x|^{4}\right),$$
and therefore we get
\begin{align}
 \left|(L-L_0)\hat G^1(s,X_s;t,\x)\right|\le C(t-s)^{2}(1+|\x|^{4})\left|X_s-x\right|.
\end{align}
So we find
\begin{align}
 \left| \Gh(0,x;t,\x)-\Gh^{(1)}(0,x;t,\x)\right|&\le C(1+|\x|^{4})\int_0^t\left( (t-s)^{2}E\left[\left|X_s-x\right|\right]+E\left[\left|X_s-x\right|^{2}\right]\right)ds
\end{align}
\noindent
The thesis then follows from estimate \eqref{eses} and integrating.
\endproof

\section{2nd-order approximation of the characteristic function}\label{app2}
For completeness we present here the formulas of the characteristic function approximation in the
general case up to the 2nd-order approximation for a process as in \eqref{eq:hetmodel} with a
local-volatility coefficient $a(t,x)$, a local default intensity $\gamma(t,x)$ and a
state-dependent measure $\nu(t,x,dz)$. We expand the coefficients around $\bar x=x$. This choice
of $\bar x$ is most common in practice and it simplifies the formulas significantly. We have
\begin{align}
\hat G^{(0)}(t,x;T,\xi)&=e^{i\xi x+(T-t)\psi(\xi)}\\
\hat G^{(1)}(t,x;T,\xi) &=\hat G^{(0)}(t,x;T,\xi)\bigg(
\frac{1}{2}i(T-t)^2\xi(i+\xi)\alpha_1\psi'(\xi)+\frac{1}{2}(T-t)^2(i+\xi)\gamma_1\psi'(\xi)\\
&-\frac{1}{2}\int_\mathbb{R}\nu_1(dz)z(T-t)^2\xi\psi'(\xi)-\frac{1}{2}\int_\mathbb{R}\nu_1(dz)(e^z-1-z)\xi\psi'(\xi)\\
&-\frac{1}{2}\int_\mathbb{R}i(e^{iz\xi}-1)(T-t)^2\psi'(\xi)\bigg)\\
\hat G^{(2)}(t,x;T,\xi)&=\hat G^{(0)}(t,x;T,\xi)\big(G^{(2)}_1(t,x;T,\xi)+G^{(2)}_2(t,x;T,\xi)+G^{(2)}_3(t,x;T,\xi)\\
&+G^{(2)}_4(t,x;T,\xi)+G^{(2)}_5(t,x;T,\xi)\big),
\end{align}
where we have defined:
\begin{align}
G^{(2)}_1(t,x;T,\xi)&=\frac{1}{2}((T-t)^2 a_2\xi(i+\xi )\psi
''(\xi)-\frac{1}{8}(T-t)^4a_1^2\xi^2(i+\xi)^2\psi'(\xi)^2\\ &-\frac{1}{6}(T-t)^3\xi
(i+\xi)(a_1^2(i+2\xi)\psi'(\xi)-2a_2\psi'(\xi)^2+a_1^2\xi(i+\xi)\psi''(\xi)),\\
G^{(2)}_2(t,x;T,\xi)&=\frac{1}{8}(T-t)^2(i+\xi)^2\gamma_1^2\psi'(\xi)^2+\frac{1}{2}(T-t)^2(1-i\xi)\gamma_2\psi''(\xi)\\
&+\frac{1}{6}(T-t)^3(i+\xi)(\gamma_1^2\psi'(\xi)-2i\gamma_2\psi'(\xi)^2+(i+\xi)\gamma_1^2\psi''(\xi)),\\
G^{(2)}_3(t,x;T,\xi)&=\frac{1}{6}(T-t)^3\xi\psi'(\xi)\int_{\mathbb{R}^2}z\nu_1(dz)+\frac{1}{3}i(T-t)^3\xi\psi'(\xi)^2\int_{\mathbb{R}}z\nu_1(dz)\\
&+\frac{1}{8}(T-t)^4\xi^2\psi'(\xi)^2\int_{\mathbb{R}^2}z\nu_1(dz)+\frac{1}{2}i\xi(T-t)^2\psi''(\xi)\int_{\mathbb{R}}z\nu_1(dz)\\
&+\frac{1}{6}(T-t)^3\xi^2\psi''(\xi)\int_{\mathbb{R}^2}z\nu_1(dz),\\
G^{(2)}_4(t,x;T,\xi)&=-\frac{1}{6}i(T-t)^3\psi'(\xi)\int_\mathbb{R}(e^{iz\xi}-1)\nu_1(dz)\int_\mathbb{R}ze^{iz\xi}\nu_1(dz)\\
&-\frac{1}{8}(T-t)^4\psi'(\xi)^2\int_{\mathbb{R}^2}(e^{iz\xi}-1)\nu_1(dz)-\frac{1}{3}(T-t)^3\psi'(\xi)\int_{\mathbb{R}}(e^{iz\xi}-1)\nu_2(dz)\\
&-\frac{1}{6}(T-t)^3\psi''(\xi)\int_{\mathbb{R}^2}(e^{iz\xi}-1)\nu_1(dz)-\frac{1}{2}(T-t)^2\psi''(\xi)\int_{\mathbb{R}}(e^{iz\xi}-1)\nu_2(dz),\\
G^{(2)}_5(t,x;T,\xi)&=\frac{1}{6}(T-t)^3\xi\psi'(\xi)\int_{\mathbb{R}^2}(e^z-1-z)\nu_1(dz)+\frac{1}{8}(T-t)^4\xi^2\psi'(\xi)^2\int_{\mathbb{R}^2}(e^z-1-z)\nu_1(dz)\\
&+\frac{1}{3}i(T-t)^3\xi\psi'(\xi)\int_{\mathbb{R}}(e^z-1-z)\nu_2(dz)+\frac{1}{6}(T-t)^3\xi^2\psi''(\xi)\int_{\mathbb{R}^2}(e^z-1-z)\nu_1(dz)\\
&+\frac{1}{2}i(T-t)^2\xi\psi''(\xi)\int_{\mathbb{R}}(e^z-1-z)\nu_2(dz).
\end{align}
Essentially $G^{(2)}_1$ corresponds to the Taylor expansion of the local volatility, $G^{(2)}_2$
results from the default function, $G^{(2)}_3$, $G^{(2)}_4$ and $G^{(2)}_5$ are related to the
state-dependent measure.


%
%

\bibliographystyle{siam}
\bibliography{Biblio}

\end{document}